\documentclass[12pt,reqno,tbtags]{amsart}

\usepackage{setspace}

\usepackage{amsmath, amsfonts, amsthm}
\usepackage{mathtools}
\usepackage{color}
\usepackage{float}
\usepackage{url}
\usepackage{comment}

\usepackage{pgfplots}
\pgfplotsset{compat=1.18}

\usepackage[dvipsnames]{xcolor}

\usepackage[top=3cm, bottom=2.5cm, left=1cm, right=1cm]{geometry}

\usepackage{newcent}
\usepackage[T1]{fontenc}
\usepackage[utf8]{inputenc}
\usepackage{babel}
\usepackage{bbm}
\usepackage{xspace}
\usepackage[hidelinks]{hyperref}
\usepackage{enumitem}
\usepackage[most]{tcolorbox}
\usepackage{xpatch}
\usepackage{stmaryrd}

\usepackage{amssymb}
\usepackage{enumitem}
\usepackage{makecell}
\usepackage{graphicx}%
\usepackage{multirow}%
\usepackage{amsmath,amssymb,amsfonts}%

\usepackage{mathrsfs}%
\usepackage[title]{appendix}%
\usepackage{xcolor}%
\usepackage{textcomp}%
\usepackage{manyfoot}%
\usepackage{booktabs}%
\usepackage{algorithm}%
\usepackage{algorithmicx}%
\usepackage{algpseudocode}%
\usepackage{listings}

\usepackage{mathtools}

\usepackage{comment}
\usepackage{caption}
\usepackage{subcaption}

\usepackage{fullpage}

\usepackage{fancyhdr}
\newcommand{\I}{\mathbbm{1}}
\newcommand{\EXP}{\mathbb{E}}
\newcommand{\Var}{\mathbb{V}}
\newcommand{\PROB}{\mathbb{P}}

\newcommand{\inlaw}{\buildrel {\mathcal L} \over =}
\newcommand{\RR}{\mathbb{R}}

\newcommand{\isdef}{\buildrel {\rm def} \over =}

\usepackage{makecell} 

\newtheorem{lemma}{Lemma}

\newtheorem{remark}{Remark}

\usepackage[numbers, sort&compress]{natbib}

\begin{document}

\title[PearsonIV]{The Pearson IV distribution: Random variate generation and applications}

\author{Luc Devroye$\dagger$}
\thanks{$\dagger$School of Computer Science, McGill University, 
		Montr\'eal, Qu\'ebec,  Canada: {lucdevroye@gmail.com}. Supported by the Natural Sciences and Engineering Research Council of Canada (NSERC) under grant number RGPIN-2024-04164}
 
 \author{Joe R. Hill$\ddagger$}
\thanks{$\ddagger$QBX Consulting, Austin, TX, USA: {joehill.qbx@gmail.com}}


\begin{abstract}
We develop uniformly fast random variate generators for
the Pearson IV distribution that can be used over the entire
range of both shape parameters and highlight some applications in a Bayesian setting.
\end{abstract}


\keywords{Random variate generation, 
Pearson IV distribution, 
Rejection method, 
Simulation, Monte Carlo method, Expected time analysis,
Log-concave distributions,
Probability inequalities}

\subjclass[2010]{65C10, 65C05, 11K45, 68U20}

\maketitle




\section{The Pearson IV distribution}\label{pearsonIV}

Undoubtedly, the most enigmatic member of Pearson's family of distributions (Pearson, 1895)\cite{Pea}
is the Pearson IV distribution, which is characterized by two
shape parameters, $a > 1/2$ and $s \in \RR$.  Its density on the real line
is given by
\begin{align}
\label{pearsonIV}
f(x) = \frac {\gamma \, e^{s \arctan (x)} }{ (1+x^2)^a},
\end{align}
where, by Legendre's duplication formula,
$$
\gamma 
\isdef  \frac { \left|  \Gamma (a - is/2) \right|^2  }{ \Gamma (a) \Gamma (a -1/2) \Gamma (1/2) }
= \frac { 4^{a-1} \left|  \Gamma (a - is/2) \right|^2  }{ \pi \Gamma (2a-1) },
$$
and $\Gamma$ is the complex gamma function.
We write $P_{a, s}$ to denote a Pearson type IV random variable with the given
parameters. 
Since $P_{a, s} \inlaw - P_{a, -s}$, we assume, without loss of generality, that $s \ge 0$.

The purpose of this paper is to propose random
variate generation algorithms that are uniformly fast over all choices of the parameters.
To the best of our knowledge, no explicit uniformly fast methods are known for this important distribution.

Section 2 recalls some facts about $P_{a,0}$, the Student-t distribution.
In the subsequent sections, we develop several generators
for the Pearson IV distribution. Some of these
require access to the normalization constant $\gamma$, which depends upon the complex gamma function.
However, we exhibit several simple algorithms that do not require explicit knowledge of $\gamma$.
We recall the two design principles for all algorithms given below:

\begin{enumerate}
\item[(i)] The generators have to be theoretically exact; no approximation of any kind is allowed.
\item[(ii)] The expected time per random variate should be uniformly bounded over all parameter choices.
\end{enumerate}

\medskip


\section{Student-t distribution}\label{Student}

The random variable $T_a$ with parameter $a > 0$ is a Student-t$(a)$ random variable if it has density
$$
\frac{1 }{ B(a/2, 1/2) \sqrt{a} ( 1+x^2 /a)^{\frac{a+1}{ 2}} },
$$
where $B$ denotes the beta function. 
First derived by Helmert \cite{Hel2, Hel} and L\"uroth \cite{Lur} and later by Pearson \cite{Pea}, it was named after Gosset (William S. Gosset \cite{Stu}) by Ronald Fisher.
It is in the Pearson IV family, as 
$$
\frac{ T_{2a-1} }{ \sqrt{2a-1} }  \inlaw P_{a,0}.
$$
We recall that
$$
T_a \inlaw \frac{ N }{ \sqrt{\frac{G_{a/2}}{a/2} }},
$$
where $N$ is standard normal, and $G_a$ denotes an independent gamma $(a)$ random variable.
Let $H_{a,b} = G_a/G_b$ be the ratio of two independent gamma random
variables (also called the beta prime distribution or beta distribution of the second kind with parameters $a$ and $b$) and let $B_{a,b}$ be a beta random variable.
From the definition of the Student-t distribution,
$$
\frac{T_a^2 }{ a} = \frac{ (1/2) N^2 }{  G_{a/2} } \inlaw \frac{  G_{1/2} }{ G_{a/2} }
\inlaw  H_{1/2, a/2}
\inlaw  B_{1/2,1/2} H_{1, a/2}
\inlaw  \sin^2 (\pi U') \left( \frac{1 }{ U^{\frac{2 }{ a}} } - 1 \right),
$$
where $U$ and $U'$ are i.i.d.\ uniform $[0,1]$ random variables.
This yields a one-liner for the Student-t distribution 
due to Bailey \cite{Bai}, also called the polar method for Student-t distribution:
\begin{align}\label{student}
T_a \inlaw \sqrt{a}\sin (2 \pi U') \sqrt{ \frac{1 }{ U^{\frac{2}{a}} } - 1 }.
\end{align}
See also Devroye \cite{D1} for variations on this polar method.
Earlier methods for the Student-t distribution include algorithms by Best \cite{Best}
and Ulrich \cite{U}.
Very simple special cases, obtainable by the inversion method, include the Cauchy law (obtained for $a=1$), for which we have $T_1 \inlaw \tan(\pi (U-1/2))$, and the $t_2$ law, for which we have
$$
T_2 \inlaw \frac{2U-1}{\sqrt{2U(1-U)}}
$$
(see, e.g. Jones \cite{LL}).

\section{Rejection from Student-t distribution}\label{rejection}

The obvious thing to try is to use the rejection method from the Student-t distribution,
for which many uniformly fast algorithms are known.
Using
$$
f(x) \le \frac{\gamma e^{s\pi/2}}{(1+x^2)^a},
$$
it suffices to generate i.i.d.\ pairs $(X,U)$, where $X = T_{2 a -1} / \sqrt{2a-1}$ is a random variable with
density proportional to $(1+x^2)^{-a}$ and $U$ is uniform on $[0,1]$,
until 
$$
U e^{s \pi/2} \le e^{s \arctan (X)},
$$
or, equivalently, to generate i.i.d.\ pairs $(X,E)$, where $E$ is standard
exponential, until 
$$
-E + s \pi/2 \le s \arctan (X).
$$

\begin{algorithm}[H]
\caption{PearsonIV generator by rejection from Student-t law}\label{Pearson1}
\begin{algorithmic}[1]
\Repeat 
\State let $E$ be an exponential random variable 
\State let $T_{2a-1}$ is a Student-t random variable with parameter $2a-1$
\State set $X \leftarrow T_{2a-1}/\sqrt{2a-1}$ 
\Until {$E \ge s(\pi/2 - \arctan (X))$} 
\State \textbf {return} $X$ 
\Comment{$X \inlaw P_{a,s}$ }
\end{algorithmic}
\end{algorithm}

As 
$$
{ \gamma e^{-s\pi/2} \over (1+x^2)^a} \le f(x) \le {\gamma e^{s\pi/2} \over (1+x^2)^a},
$$
it is easy to see that the expected number of iterations is 
at least $(1/2) e^{\pi s/2}$ (and at most $e^{\pi s}$),
which is uniformly bounded for all $a > 1/2$ and $ s < c$ for some fixed constant $ c$.
As soon as $s \ge 10$, or something of that order of magnitude, this 
simple method is unfeasible.


\section{The Pearson IV density}\label{pearson4}

The Pearson IV density \eqref{pearsonIV} is unimodal and has mean
$$
\mu = \frac{s}{2(a-1)}
$$
for $a > 1$, and variance
$$
\sigma^2 = \frac{s^2 + 4(a-1)^2}{4 (a-1)^2 (2a-3)}
$$
for $a > 3/2$.
It has a unique mode at
$$
m = \frac{s}{2a}.
$$
The log-density is
$$
\log f(x) = \log (\gamma) + s \arctan (x) -a \log( 1+x^2 ), x \in \RR.
$$
The derivatives of $\log f$ are
\begin{align}
 (\log f)'(x) &= \frac{s-2ax}{1+x^2}, \\
 (\log f)''(x) &= \frac{2ax^2 -2sx -2a }{(1+x^2)^2} .   
\end{align}
This shows that the Pearson IV density is log-concave on the interval defined by
$$
| x - m | \le D \isdef \sqrt{ 1+m^2 },
$$
and log-convex outside that interval.


\section{The arctan-mapped density}\label{arctan}

For $a \ge 1$, Exercise 1 on page 308 in Devroye \cite{D5} notes that $\arctan (P_{a,s})$ has a log-concave density
on $[-\pi/2, \pi/2]$ given by
\begin{equation}\label{atan}
h(y) =  
\begin{cases}
\gamma e^{sy} ( \cos^2 (y) )^{a-1} & \text{if} ~|y| \le \frac{\pi }{ 2}, \\
0 & \text{else.} 
\end{cases}
\end{equation}
The function $g = \log (h)$ has the following derivatives, all decreasing in $y$ on $[0,\pi/2]$:
\begin{align}\label{derivatives}
g' (y) &= s - 2(a-1) \tan (y) ,\\
g'' (y) &= - \frac{2 (a-1)} {\cos^2 (y)}, \\
g''' (y) &= - \frac{4 (a-1) \tan (y)} {\cos^2 (y)} ,\\
g'''' (y) &= - \frac{4 (a-1) ( 1 + 2 \sin^2 (y)) } {\cos^4 (y)}.
\end{align}
The modes of $h$ and $g$ occur at
$$
m = \arctan (\beta),
$$
where we set $\beta = s/(2(a-1))$, noting that for $a=1$, we have $m = \pi/2$.
The mean and variance $\Var \{ X \}$ of a random variable $X$
with density \eqref{atan} can be expressed as a function of the complex
digamma function and the trigamma function $\psi_1$, e.g.,
\begin{align}\label{variance}
\Var \{ X \} 
= \frac{1}{4} \left( \psi_1 (a+ is/2) + \psi_1 (a- is/2) \right)
\in \left[ \frac{2a}{s^2 + 4a^2} , \frac{2(a+1)}{s^2 + 4a^2}\right]
\end{align}
\cite{Barndorff1982}. 
For any random variable $X$ with a unimodal density $h$ with mode at $m$, Dharma\-dhikari and Joag-Dev \cite{Dharmadhikari1982}\cite{Dharmadhikari1988}\cite{Abu-Bakut2011} showed that
$$
\frac{1}{2e } \le h(m)\sqrt{\Var\{X\}}.
$$
For log-concave densities, Fradelizi \cite{Fradelizi1997} improved this:
\begin{align}\label{fradelizi}
\frac{1}{\sqrt{12} } \le h(m)\sqrt{\Var\{X\}} \le \frac{1}{\sqrt{2}}\,.
\end{align}
The upper bound improves over an earlier result by Ibragimov \cite{ibragimov1956unimodal}, who showed that $h(m)\sqrt{\Var\{X\}} \le 1$.
Therefore,
\begin{align}\label{Lbound}
h(m) 
\ge 
\sqrt{\frac{4a^2 + s^2}{24(a+1)}}.
\end{align}
Several rejection methods could be used at this point.
\begin{enumerate}[label=(\roman*)]
\item
By an inequality due to Devroye \cite{D6}, we have
\begin{align}\label{devroye}
h(x) \le \min \left( h(m),  h(m) \exp(1 - h(m)|x-m| ) \right) .
\end{align}
Using this would yield an algorithm taking an expected number of ite\-ra\-tions equal
to 4, but it would require access to the normalization constant $\gamma$.
\item
We apply the bound $h(m) \ge L$ in \eqref{devroye} to obtain
\begin{align}\label{devroye2}
h(x) &\le \min \left( h(m),  h(m) \exp(1 - L|x-m| ) \right).
\end{align}
This avoids computing $\gamma$.  The expected number of iterations
becomes 
$$
4 \frac{h(m)}{L}.
$$
Candidates for $L$ include $(\gamma^-/\gamma) h(m)$ (see \eqref{gammabounds} below) and \eqref{Lbound}, which is based on Fradelizi's inequality \cite{Fradelizi1997}).  The former replacement pushes the expected numbers of iterations up slightly to $4\gamma/ \gamma^- \le 4\gamma^+/\gamma^-$, which is uniformly bounded over all $a \ge 1$ and $s \in \RR$.
\item
A custom-designed lower bound for $h(m)$. 
\end{enumerate}

\bigskip    

\begin{figure}[H]
    \centering

\begin{tikzpicture}
\begin{axis}[
    width=12cm, height=8cm,
    samples=100,
    domain=-1.5708:1.5708,
    xlabel={$x$}, ylabel={$f(x)$},
    grid=major,
    legend pos=outer north east,
    legend style={font=\tiny}
]
\def\consts{{1.1641, 1.1016, 0.9346, 0.7137, 0.4932, 0.3107, 0.1798, 0.0964, 0.0482, 0.0226, 0.0101}}
\foreach \s in {-10,-9,...,10} {
    \pgfmathsetmacro{\absS}{abs(\s)}
    \pgfmathsetmacro{\c}{\consts[\absS]}  
    \addplot+ [mark=none, thick] {
        \c * exp(\s*x) * (cos(deg(x))^8)
    };
    \addlegendentryexpanded{s = \s}   
}
\end{axis}
\end{tikzpicture}
\caption{Normalized density $\gamma e^{sx} \cos^{2a-2}(x)$ for $a=5$ and various values of $s$.}
\end{figure}
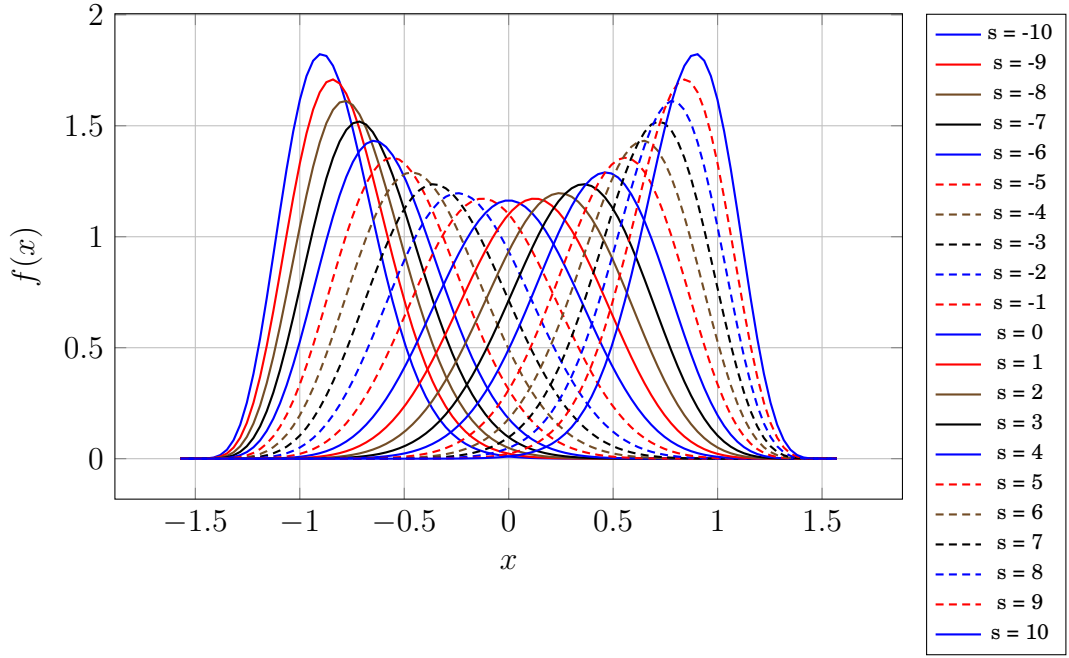

\begin{remark}
    \textsc{the skewed cauchy family.}
When $a=1$, $h(y) = \gamma \exp(sy)$, so that a random variate with density $h$ simply is
\begin{align}\label{skewedcauchy}
W = 
\begin{cases}
\frac{1 }{ s} \log \left( e^{-\frac{\pi s }{ 2}} + U \left( e^{\frac{\pi s }{ 2}} - e^{-\frac{\pi s }{ 2}} \right) \right) & \text{if}~s > 0, \\
\pi (U-1/2) & \text{if}~s = 0, 
\end{cases}
\end{align}
where $U$ is uniform on $[0,1]$. Therefore, $P_{1,s} \inlaw \tan (W)$. 
We also rediscover the standard method
for generating Cauchy random variables: $P_{1,0} \inlaw \tan (\pi (U-1/2))$.
The family of distributions  $P_{1,s}$  will be called the skewed Cauchy family. 
\end{remark}

In the remainder of this section, we assume that $a \ge 1$. Note that
$$
h(m) 
= \gamma e^{s \arctan (\beta) } \left(\frac{1 }{ 1+ \beta^2} \right)^{a-1}
= \gamma \left( \frac{ e^{2 \beta \arctan (\beta)}}  { 1+ \beta^2} \right)^{a-1}.
$$
The universal method for log-concave densities
from Devroye (1984)\cite{D4}, adap\-ted here for use with  the inequality 
based on \eqref{devroye2} is given below

\begin{algorithm}[H]
\caption{Universal log-concave method: $a \ge 1$}\label{LC}
\begin{algorithmic}[1]
\State let $m$ be the location of the mode of the log-concave density $h$ \eqref{atan}
\State let $L$ be such that $h(m) \ge L$
\Repeat  
\State let $V$ be uniform on $[-2,2]$ 
\If {$V<-1$} 
\State replace $V$ by $-1 + \log(V+2)$ 
\ElsIf {$V>1$} 
\State replace $V$  by $1 - \log(V-1)$ 
\EndIf
\State $Y \gets m + V/L$ 
\State let $U$ be uniform on $[0,1]$ 
\Until {$U \min ( 1 ,  \exp (1-L|Y-m|) ) \le h(Y)/h(m)$} 
\State \textbf {return} $Y$ 
\Comment        {$Y$ has density $h$ }
\State {}
\Comment        {$\tan(Y) \inlaw P_{a,s}$ if $h$ is as in \eqref{atan}}
\end{algorithmic}
\end{algorithm}

\begin{remark}
    \textsc{adaptive methods.}  The method given here has a uniformly bounded
time and is useful when the parameters vary in an application. For fixed parameters, several adaptive methods make the method more efficient as more random variates
are generated. See, e.g., Gilks \cite{G1}, Gilks and Wild \cite{GW1, GW2} and Gilks, Best and Tan \cite{GBT}. $\square$
\end{remark}

\begin{remark}
    \textsc{references.}  Additional references on random variate generation for 
log-concave laws include H\"ormann, Leydold and Derflinger \cite{HLD}, Leydold and H\"ormann \cite{LH, LH2}
and Devroye \cite{D5}.  For an implementation of the algo\-rithm in this section, see Heinrich \cite{Hei}. $\square$
\end{remark}

\begin{remark}
\textsc{the universal algorithm without access to the norma\-li\-zation constant.}
With $L= h(m)$, we know that the expected number of iterations in algorithm \ref{LC}
is 4. In Lemma \ref{gamma} below, we show 
that 
$$
\gamma^- \le \gamma \le \gamma^+,
$$
where $\gamma^-$ and $\gamma^+$ are explicit functions of the parameters, $a$ and $s$.
Thus, setting
$$
\Delta = \frac{h(m)}{\gamma}
=   e^{sm} ( \cos^2 (m))^{a-1},
$$
we have $h(m) \ge L \isdef \gamma^- \Delta$. If this value of $L$ is used in algorithm {LC}, then
the expected number of iterations
is the  integral of the bounding functions, or
$$
4  \frac{\gamma }{ \gamma^-} \le
4  \frac{ \gamma^+ }{ \gamma^-} .
$$
With the choices given in Lemma \ref{gamma} below, we see that the expected number of iterations is
$$
\le 4 \times \left( 
\frac
{ 1 + \frac{3 }{ 2 \pi^2 \sqrt{a^2 + (s/2)^2}} }
{ 1 - \frac{3 }{ 2 \pi^2 \sqrt{a^2 + (s/2)^2} } } 
\right)^2
\times \sqrt{\frac{ (a+ 0.177)(a+0.677)  }{ (a+ 1/6)(a+2/3) }}.
$$
Uniformly over all $a \ge 1$ and $s \ge 0$, this does not exceed
$$
4 \times \left( \frac{ 2\pi^2 +3  }{ 2 \pi^2 - 3 } \right)^2 \times \sqrt{\frac{ 1.177 \times 1.677 \times 18  }{ 35 }}
\approx 5.34.
$$
Note, though, that as $a \to \infty$, the expected number of iterations tends to 4, since the
inequalities for the gamma function get tighter.
\end{remark}

\section{Rejection from an exponential for the arctan-mapped density}\label{gaussian}

For $a \ge 1$, we have $h(y) \le \gamma e^{sy}$  on $[-\pi/2,\pi/2]$. This leads directly to the following rejection algorithm:

\begin{algorithm}[H]
\caption{PearsonIV generator by rejection from the exponential density when $a \ge 1$.}\label{expo-rejection}
\begin{algorithmic}[1]
\Repeat 
\State let $V$ be uniform on $[0,1]$
\State let $Y \leftarrow \frac{1}{s} \log \left( V e^{s\pi/2}  + (1-V) e^{-s\pi/2} \right)$
\State let $U$ be uniform on $[0,1]$
\Until { $U   \le ( \cos (Y) )^{2a-2} $} 
\State \textbf {return} $X \leftarrow \tan(Y)$  
\Comment{$X \inlaw P_{a,s}$ }
\end{algorithmic}
\end{algorithm}

If we write $\gamma = \gamma (a,s)$ to make the dependence of the normalization constant on the
parameters explicit, then it is easily seen that the expected number of iterations of the 
algorithm \ref{expo-rejection} is
$$
\frac{\gamma (a,s)}{\gamma (1,s)},
$$
which, by Lemma \ref{gammabounds}, for fixed $a \ge 1$, increases in proportion to $s^{2a-2}$ as $s \uparrow \infty$. Thus, algorithm \ref{expo-rejection} has uniformly bounded time for $1 \le a \le a^*$,
$0 \le s \le s^*$. where $a^*$ and $s^*$ are small constants.  It could be useful for $1 \le a \le 3$, $0 \le s \le 3$.

\section{Gaussian domination for the arctan-mapped density}\label{gaussian}

We can derive upper bounds for the arctan-mapped density $h$ (see \eqref{atan}) based on
the derivatives of $g=\log h$ (see \eqref{derivatives}) and standard Taylor series bounds:
\begin{align}
g(y) 
&\le 
\begin{cases}
g(m) + \frac{(y-m)^2}{2} g''(m), & y \ge m , \cr
g(m) + \frac{(y-m)^2}{2} g''(m) - \frac{(y-m)^3}{6} g'''(m), & y \le m, \cr
\end{cases}  \\
&\le 
\begin{cases}
g(m) + \frac{(y-m)^2}{2} g''(m), & y \ge m , \cr
g(m) + \frac{(y-m)^2}{4} g''(m),  & m - D \le y \le m, \cr
\end{cases}\label{gaussian2}
\end{align}
whenever 
$$
\frac {D |g'''(m)|}{6} \le \frac{ |g''(m)|}{4}.
$$
This is satisfied if we pick $D$ such that $D \beta \le 3/4$.
For example, if we set $D=\pi$, then for all $m$ satisfying
\begin{align}\label{condition}
\tan (m) \le \frac{3}{4 \pi} \approx 0.2387324,
\end{align}
we have
\begin{align}\label{gaussian3}
h(y) \le 
h(m) e^{-\tau^2 \frac{(y-m)^2}{2}}, y \in \RR,
\end{align}
where
\begin{align*}
\tau \isdef \sqrt {|g''(m)|/2} = \sqrt{(a-1)(1+\beta^2)}.
\end{align*}
A condition equivalent to \eqref{condition}
is $|m| \le \arctan \left( \frac{3}{4\pi} \right) \approx 0.2344$. 
Algorithm  \ref{simplegaussian} uses rejection either from the Gaussian implied by \eqref{gaussian3} or rejection from the uniform density on $[-\pi/2,\pi/2]$.

\begin{lemma} 
The expected number of iterations taken by algorithm \ref{simplegaussian} is not more than
$$
\sqrt{2 \pi + \pi^2}  \approx 4.02,
$$
uniformly over all values of the parameters with $a \ge 1$ and $|m| \le \arctan \left( \frac{3}{4\pi} \right) \approx 0.2344$.
\end{lemma}

\begin{proof}
If we were to use \eqref{gaussian2}, then the  expected number of iterations before halting would be
$$ 
\frac{ \sqrt{2 \pi}  h(m) } { \tau  } .
$$
However, if we bound $h$ by $h(m)$ and use rejection from the uniform density on $[-\pi/2,\pi/2]$, then the expected number of iterations before halting
would be $\pi h(m)$.  Taking  the best of both leads to the cut-off
value at $\tau = \sqrt{2/\pi }$. 
Using \eqref{variance} and Fradelizi's inequality \eqref{fradelizi}, we have
\begin{align*}
h(m) \le \sqrt{ a + \frac{s^2}{4a}} .
\end{align*}
Assume first $\tau \le \sqrt{2/\pi }$. Then 
$a-1 \le 2/\pi$ and $s^2/(4(a-1)) \le 2/\pi$, which implies that
$$
\pi h(m) \le \pi \sqrt{1 + 4/\pi} = \sqrt{\pi^2 + 4}.
$$
If, on the other hand, $\tau \ge \sqrt{2/\pi }$, then
\begin{align*}
\frac{ \sqrt{2 \pi}  h(m) } { \tau  }
&\le
\frac{ \sqrt{2 \pi}  \sqrt{a +\frac{s^2}{4a} } } { \tau  }
\le
\frac{ \sqrt{2 \pi}  \sqrt{1 +(a-1) +\frac{s^2}{4(a-1)} } } { \tau  }\\
&=
\frac{ \sqrt{2 \pi}  \sqrt{1 + \tau^2 } } { \tau  }
=
\sqrt{2 \pi}  \sqrt{1 + \frac{1}{\tau^2} } 
\le 
\sqrt{2 \pi}  \sqrt{1 + \frac{\pi}{2} } 
=
\sqrt{2 \pi + \pi^2} .
\end{align*}
\end{proof}

\begin{figure}[H]
    \centering\label{pinkoo}
    \includegraphics[width=\linewidth]{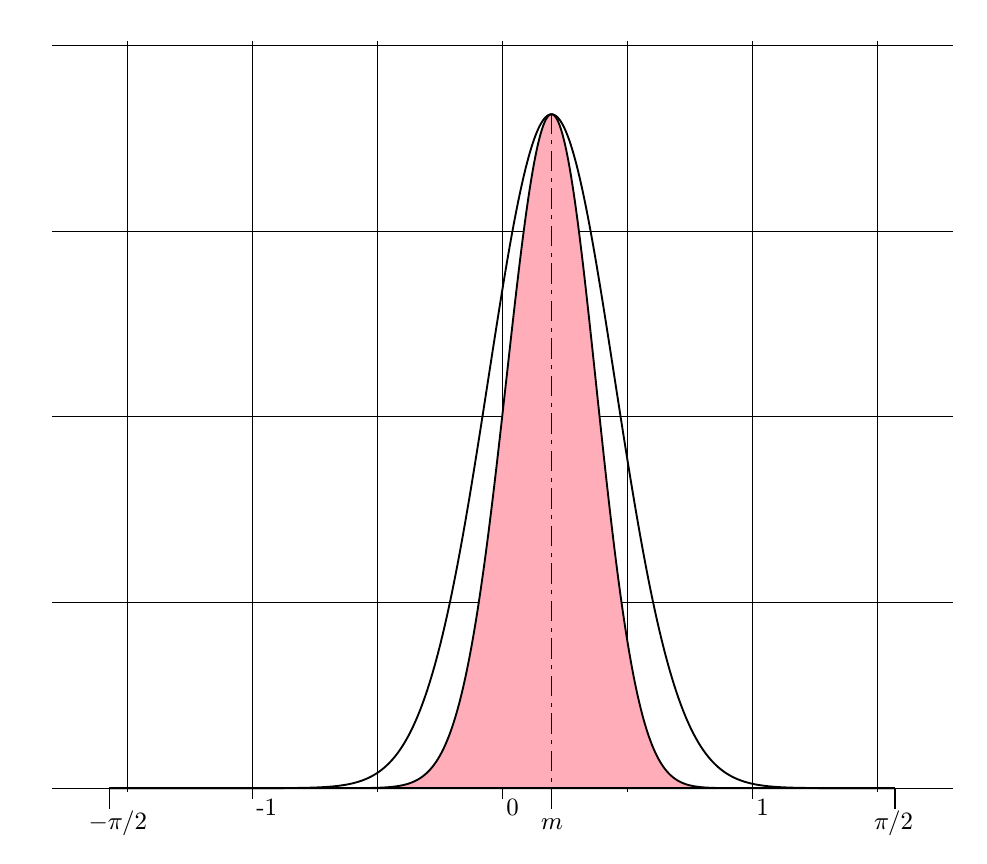}
    \caption{The arctan-mapped Pearson IV density is shown in pink, together with the normal envelopes for $s=6$, $a=16$.}
\end{figure}

\begin{algorithm}[H]
\caption{PearsonIV generator for $a > 1$ by rejection from the normal density when $\tan (m) \le 3/(4 \pi) \approx 0.2387324$, or equivalently, when $m \le  \arctan(3/(4 \pi)) \approx 0.2344$.}\label{simplegaussian}
\begin{algorithmic}[1]
\State let $\beta = s/(2(a-1))$,  $m = \arctan(\beta)$, $\tau = \sqrt{(a-1)(1+\beta^2 )}$
\If{$\tau \ge \sqrt{2/\pi}$}
\Repeat 
\State let $N$ be a standard Gaussian random variable 
\State let $U$ be uniform on $[0,1]$
\State set $Y \leftarrow  m + N/\tau$
\Until {$U h(m) e^{-\tau^2  \frac{(Y-m)^2}{2}}  \le h(Y)$} 
\State
\Comment{the acceptance condition is equivalent to $U e^{-N^2 /2} \le h(Y)/h(m)$}
\Else
\Repeat 
\State let $Y$ be uniform on $[-\pi/2, \pi/2]$ 
\State let $U$ be uniform on $[0,1]$
\Until {$U   \le h(Y)/h(m)$} 
\EndIf
\State \textbf {return} $X \leftarrow \tan(Y)$  
\Comment{$X \inlaw P_{a,s}$ }
\end{algorithmic}
\end{algorithm}
\section{Bounding the normalization constant}\label{gammabounds}

The gamma function is defined for complex $z=x+iy$ with $x> 0$ by Euler's integral
$$
\Gamma (z) = \int_0^\infty t^{z-1} e^{-t} \, dt.
$$
Explicit inequalities for Euler's gamma function with real argument are often tied to Stirling's
approximation (Stirling \cite{Stir}).  A prime example is Robbins's upper and lower bound (Robbins \cite{Ro}).
Olver et al.\ \cite{Ol} summarize most of the well-known bounds.

\begin{lemma}\label{batir} (Batir \cite{Bat}).
\begin{align}\label{batir1}
\sqrt{2e}  \left( \frac{ x + 1/2 }{ e} \right)^{x+1/2} \le \Gamma (1+x) \le \sqrt{2\pi} \left( \frac{ x + 1/2 }{ e} \right)^{x+1/2} , x > 0,
\end{align}
and
\begin{align}\label{batir2}
\sqrt{2\pi (x+ 1/6)}  \left( \frac{ x }{ e} \right)^{x}\le \Gamma (1+x) \le \sqrt{2\pi \left(x+ \frac{e^2 }{ 2\pi} -1 \right) } \left( \frac{ x }{ e} \right)^{x} , x \ge 1.
\end{align}
\end{lemma}

\begin{lemma}\label{boyd} (Boyd \cite{Boyd}; see also (5.11.11) in Olver et al. \cite{Ol}).
For $\Gamma(z)$ with $z=x+iy$, $x > 0$, 
\begin{align}\label{boyd1}
\left| \Gamma (z) \right| 
= \left| \sqrt{\frac{2 \pi }{ z} }\left( \frac{ z }{ e } \right)^{z} \right| 
\times \left| 1 + R(z) \right|,
\end{align}
where
\begin{align}\label{boyd2}
| R(z) | \le \frac{3 }{ 2 \pi^2 |z|}.
\end{align}
\end{lemma}

\begin{lemma}\label{gamma}
Let $a \ge 1$ and $s \ge 0$. We have $\gamma^- \le \gamma \le \gamma^+$, where
\begin{align}\label{gammabounds}
\gamma^* 
=& 
\frac{
(a-1/2) \left( 1 + (s/2a)^2  \right)^{a-1/2} e^{-s \arctan (s/2a) } 
}{
\sqrt{\pi/e} \left(  1+1/2a  \right)^a  \sqrt{a} } , \\
\gamma^+ &= \gamma^* \times \frac{ \left (1 + \frac{3 }{ 2 \pi^2 \sqrt{a^2 + (s/2)^2 }} \right)^2  }{ \sqrt{ 1 + \frac{1 }{ 6a}}
\sqrt{ 1 + \frac{1 }{ 6(a+1/2)}} } ,\\
\gamma^- &= \gamma^* \times \frac{ \left (1 - \frac{3 }{ 2 \pi^2 \sqrt{a^2 + (s/2)^2 }} \right)^2  }
{ \sqrt{ 1 + \frac{0.177 }{ a}}  \sqrt{ 1 + \frac{0.177 }{ a+1/2}} }. 
\end{align}
\end{lemma}

\begin{proof}
Combine Batir's inequality \eqref{batir2} with Boyd's bounds \eqref{boyd1} \eqref{boyd2} after re\-wri\-ting $\gamma$ as follows: 
$$
\gamma 
\isdef  \frac{ \left|  \Gamma (a - is/2) \right|^2  }{ \Gamma (a) \Gamma (a -1/2) \Gamma (1/2) }
=  \frac{ (a^2 - 1/4) a \left| \Gamma (a - is/2) \right|^2  }{ \Gamma (a+1) \Gamma (a+3/2) \sqrt{\pi} },
$$
and noting that
$e^2/(2\pi) -1 < 0.177$ and
\begin{equation}
\left|  \sqrt{\frac{2 \pi }{ z}} \left( \frac{ z }{ e } \right)^{z} \right|  
= \sqrt{ \frac{ 2 \pi }{ \sqrt{x^2 + y^2} }} \left( \frac{ \sqrt{x^2 + y^2} }{ e } \right)^{x} e^{-y \arctan (y/x) }.
\end{equation}
\end{proof}


\section{Symmetrization for parameter values $a \le 1$}\label{symmetric}

Finally, we develop a generator that is uniformly fast for $a \in (1/2,1], s \in \RR$.
As $P_{a,s} \inlaw -P_{a, -s}$, symmetrization may be helpful.
Define the symmetric density
$$
g(x) = \frac{f(x) + f(-x) }{ 2} = \frac{\gamma  \cosh ( s \arctan (x) ) }{ (1+x^2)^a }.
$$
Setting $Y=\arctan (X)$, where $X$ has density $g$ on $\RR$ yields the following 
symmetric density on $(-\pi/2, \pi/2)$:
$$
h(y) =  \gamma \cosh (sy) ( \cos^2 (y) )^{a-1}.
$$
Consider next the random variable $Z = \pi/2 - |Y|$
with density
\begin{equation}\label{eta}
\eta (z) =  2\gamma \cosh (s (\pi/2 -z)) ( \sin (z) )^{2(a-1)}, 0 \le z \le \pi/2.
\end{equation}

For $a \in (1/2,1]$,  the density \eqref{eta}
is decreasing and has an infinite peak at the origin unless $a=1$.
Most of its mass is near zero, and thus, we will attempt rejection using the bound
\begin{align*} 
\eta (z) &\le 2\gamma e^{s\pi/2} e^{-sz} (2z/\pi)^{2(a-1)}  \\
&\le 2\gamma (2/\pi)^{2(a-1)} e^{s\pi/2} e^{-sz} z^{2(a-1)} ,  0 < z \le \pi/2.
\end{align*}

Assume first that $s \ge 1$.
Introduce a uniform $[0,1]$ random variate $U$ and an independent  gamma $(2a-1)$
random variate $G_{2a-1}$.
We apply rejection from the gamma distribution by generating
independent pairs $(Z,U) = (G_{2a-1}/s , U)$ until for the first time,
$Z \le \pi/2$ and
$$
U  (2Z/\pi)^{2(a-1)} \le (\sin (Z))^{2(a-1)},
$$
or, equivalently,
$$
U \le \left( \frac{ 2Z }{ \pi \sin (Z)} \right)^{2(1-a)}.
$$
The returned random variable $Z$ has density $\eta$ given in \eqref{eta}.
The probability of acceptance is thus at least

\begin{align*} 
\EXP &\left\{ \left( \frac{ 2G_{2a-1}/s }{ \pi \sin (G_{2a-1}/s)} \right)^{2(1-a)}  \I_{G_{2a-1}/s \le \pi/2} \right\} \\
&\ge \EXP \left\{ \frac{ 2G_{2a-1}/s }{ \pi \sin (G_{2a-1}/s)}  \I_{G_{2a-1}/s \le \pi/2} \right\} \\
&\ge \frac{2 }{ \pi}  \PROB \left\{ G_{2a-1}/s \le \pi/2 \right\}, 
\end{align*}
where $\I$ is the indicator function. By Markov's inequality, the probability in this expression is at least
$$
1 - \frac{2 \EXP \{ \frac{G_{2a-1} }{ s} \} }{  \pi} = 1 - \frac{2 (2a-1) }{ s \pi}
\ge 1 - \frac{2 }{ s \pi} 
\ge 1 - \frac{2 }{ \pi} 
$$
uniformly for all $s \ge 1, a \le 1$. In this range, the rejection algorithm's expected number of iterations is at most 
$$
\frac{ \pi^2 }{ 2\pi - 4 }.
$$
Having generated $Z$ with density \eqref{eta}, we need to set 
$X = \tan ( S(\pi/2 -Z) )$, where $S$ is a random sign to obtain 
a random variate with the symmetrized Student-t density $g$.
Finally, a random variate with the Pearson IV density $f$ is obtained as

\begin{equation*}
\begin{cases}
X & \text{with probability } \frac{f(X) }{ f(X)+f(-X)} 
= \frac{ e^{s \arctan (X)}  }{ e^{s \arctan (X)} + e^{-s \arctan (X)}},  \\
-X & \text{else.} 
\end{cases}
\end{equation*}

Next, assume that $0 \le s \le 1$. Then
$$
\eta (z) \le 2\gamma (2/\pi)^{2(a-1)} e^{s\pi/2} z^{2(a-1)} ,  0 < z \le \pi/2.
$$
Thus, we can generate random pairs 
$$
(Z,U) = \left( V^{\frac{1}{2a-1}} \frac{\pi}{2} , U  \right)
$$
where $U, V$ are independent uniform random variates, until for the first time
$$
U \le e^{-s Z}.
$$
The random variate $Z$ has density \eqref{eta}.
The expected number of iterations is
$$
\EXP \{ e^{sZ} \} \le e^{s\pi/2} \le e^{\pi/2}.
$$

We combine the algorithms below.

\begin{algorithm}[H]
\caption{PearsonIV generator for parameter $a \in (1/2, 1].$}\label{small-a}
\begin{algorithmic}[1]
\If {$s \ge 1$}
\Repeat
\State generate $Z = G_{2a-1}/s$ and $U$ uniformly on $[0,1]$ 
\Until {$Z < \pi/2$ and $U \le \left( \frac{ 2Z }{ \pi \sin (Z)} \right)^{2(1-a)}$} 
\Else
\Repeat
\State generate $V$ and $U$ uniformly on $[0,1]$ 
\State set $Z \leftarrow  V^{\frac{1}{2a-1}} \frac{\pi}{2}$
\Until {$U \le e^{-sZ}$} 
\EndIf
\State $Y \gets  S(\pi/2 -Z) $, where $S$ is a random sign 
\State $X \gets \tan ( Y ) $ 
\State with probability $e^{-sY}/(e^{-sY} + e^{sY})$, $X \gets -X$ 
\State \textbf{return} $X$ 
\Comment {$X \inlaw P_{a,s}$}
\end{algorithmic}
\end{algorithm}

\begin{remark}
\textsc{gamma random variates.}
For uniformly fast gamma random variates, we refer to the surveys
in Devroye \cite{D5} and Luengo \cite{Lu}. In terms of the rejection constant,
the method of Marsaglia and Tsang \cite{Marsaglia+Tsang}
is highly recommended. Many simulation studies confirm
that the method of Schmeiser and Lal \cite{SL} is quite
competitive if the gamma parameter is at least one. Xi, Tan and Liu
\cite{Xi} suggested generating $\log (G_a)$ instead. As $\log (G_a)$
has a log-concave density for all values of $a > 0$, a uniformly
fast generator is quite easily obtained either by the universal
method of Devroye \cite{D4} or a specialized algorithm as developed, e.g.,
in Devroye \cite{D8}.
For algorithm \ref{small-a}, with a gamma parameter less than one, we recommend the one-liner recently developed by Greaves \cite{greaves2026extended}. $\square$
\end{remark}

\section{Putting things together for the Pearson IV distribution.}

We conclude by 
providing an overview of the methods developed above, which are all uniformly fast over the respective ranges of the parameters specified below. Algorithms \ref{LC} and \ref{small-a}, taken together, cover the entire parameter space.

\begin{enumerate}
\item[(i)] For $a \ge 1$ and all $s$, one can use the universal log-concave generator (algorithm \ref{LC}).
\item[(ii)] For $|s| \le 5$ and all $a$, one can use rejection from the Student-t density (algorithm \ref{Pearson1}). 
\item[(iii)] When both $|s|\le 3$ and $1 \le a \le 3$, one can use rejection from an exponential distribution (algorithm \ref{expo-rejection}). 
\item[(iv)] For $a >1$ and $|s| \le 3 (a-1)/ (2 \pi)$, one can use a simple rejection algorithm based on a normal envelope (algorithm \ref{simplegaussian}).
\item[(v)] For $a \in (1/2, 1]$, one can use rejection from the gamma density after symmetrizing the Pearson distribution (algorithm \ref{small-a}).
\item[(vi)] For $a = 1$, the skewed Cauchy density, there is a simple one-liner \eqref{skewedcauchy}.
\item[(vii)] For $s = 0$, the Student-t law, there is a simple one-liner \eqref{student}.
\end{enumerate}

\section{Simulations.}

We verified by simulation that all proposed algorithms generate variates consistent with the Pearson~IV distribution. All timings were obtained on an Intel Xeon Gold 6234 CPU @ 3.30\, GHz (64-bit, x86-64 GNU/Linux) and are reported in microseconds per variate, averaged over $10^5$ samples. No constants were precomputed for any of the algorithms.
For calibration, generating an exponential variate via $\ln(1/U)$, where $U \inlaw \mathrm{Uniform}(0,1)$, required $0.17\,\mu\text{s}$. A Student-$t$ variate using Bailey's method required $0.74\,\mu\text{s}$, while the skewed Cauchy method (valid for $P_{1,s}$, any $s$) required $2.11\,\mu\text{s}$ per variate.
Algorithm~5, for $a \in (1/2,1]$, achieved timings between $5.3\,\mu\text{s}$ and $6.3\,\mu\text{s}$ across all admissible $(a,s)$.
The table below compares Algorithms~1--4 for various parameter combi\-nations $(a,s)$, where Algorithm~4 uses $L$ as in \eqref{Lbound}. The universal Algorithm~2 demonstrates stable performance across the entire parameter space, although it is generally slower than Algo\-rithms~3 and~4, which are applicable only on restricted subsets of the parameter domain.

\begin{table}[H]
\centering
\renewcommand{\arraystretch}{2} 
\begin{tabular}{c|c|c|c|}
& $a=1$ & $a=3$ & $a=9$ \\
\hline
\makecell{$s=1$ } &
\makecell{5.9 \\ \textcolor{RoyalPurple}{10.9} \\ \textcolor{Magenta}{2.8} }& 
\makecell{7.7 \\ \textcolor{RoyalPurple}{14} \\ \textcolor{Magenta}{6.3} \\ \textcolor{BurntOrange}{7.6} } & 
\makecell{8.1 \\ \textcolor{RoyalPurple}{15.9} \\ \textcolor{Magenta}{11.4} \\ \textcolor{BurntOrange}{7.5} } \\
\hline
\makecell{$s=3$ } &
\makecell{16.1 \\ \textcolor{RoyalPurple}{23.1} \\ \textcolor{Magenta}{2.8} }& 
\makecell{79.6 \\ \textcolor{RoyalPurple}{12} \\ \textcolor{Magenta}{20} \\ \textcolor{BurntOrange}{7.5} } & 
\makecell{143.1 \\ \textcolor{RoyalPurple}{15} \\   \textcolor{BurntOrange}{7.4} } \\
\hline
\makecell{$s=9$ } &
\makecell{ \textcolor{RoyalPurple}{18.1} \\ \textcolor{Magenta}{2.8} }& 
\makecell{ \textcolor{RoyalPurple}{7.1}   } & 
\makecell{ \textcolor{RoyalPurple}{9.6}  } \\
\hline
\end{tabular}
\caption{Timing in microseconds per random variate. The color scheme: algorithm 1, \textcolor{RoyalPurple}{algorithm 2}, \textcolor{Magenta}{algorithm 3}, \textcolor{BurntOrange}{algorithm 4}.}
\end{table}

\section{A statistical model involving the Pearson IV family}

In this section, we describe a Bayesian statistical model involving the Pearson IV family of distributions. For convenience, we define the following notation.
If $X$ is a Pearson IV random variable with parameters $a = m_0/2 + 1$ and $s = m_0\, \mu_0$, with density
$$
f_X(x) = K(\mu_0, m_0) \exp\{m_0\, \mu_0 \arctan(x) - (m_0/2 + 1) \log(x^2 + 1)\},
\ x \in \mathbb{R},
$$
having normalizing constant
$$
K(\mu_0, m_0) = \frac{4^{m_0/2}
\, \left|\, \Gamma\left(\frac{m_0}{2} + 1 - i\, \frac{m_0\, \mu_0}{2}\right)\, \right|^2 }
{\pi\, \Gamma(m_0)},
$$
then we write
$$
X \sim \text{Pearson IV}(\mu_0, m_0),
$$
which has mean and variance
$$
\EXP \{X \} = \mu_0 \ \text{ and }\ \Var \{X \} = (\mu_0^2 + 1)/(m_0 - 1).
$$

If $X$ is a random variable having density belonging to the natural exponential family generated by the convolved hyperbolic secant distribution with parameters $\mu$ and $n$, with density
$$
f_X(x) = H(x, n)\, \exp\{x\, \arctan(\mu) - \frac{n}{2}\, \log(\mu^2 + 1)\},
\ x \in \mathbb{R},
$$
where $H(x, n)$ is the density of a convolved hyperbolic secant distribution,
$$
H(x, n) = \frac{2^{n - 2}}{\pi\, \Gamma(n)} \left|\, \Gamma\left(\frac{n}{2} + i\, \frac{x}{2}\right)\, \right|^2,
$$
then we write
$$
X \sim \textsc{nef-chs}(\mu, n),
$$
which has mean and variance
$$
\EXP \{X \} = n\,\mu \ \text{ and }\ \Var\{X\} = n\, (\mu^2 + 1).
$$
Note that if $\bar{X} = X/n$, then $\EXP \{ \bar{X} \} = \mu \ \text{ and }\ \Var\{ \bar{X} \} = (\mu^2 + 1)/n$.
See Morris \cite{M1, M2} for additional properties of the \textsc{nef-chs} and the other five natural exponential families with quadratic variance functions. Devroye \cite{D3} defines a uniformly fast and exact algorithm for generating \textsc{nef-chs} variates.

Given these definitions, we assume that a current observation $Y$ and a future observation $Z$ are conditionally independent given an unknown mean parameter $\mu \in \mathbf{R}$ and a known sample size $n \geq 1$, with common density belonging to the \textsc{nef-chs} sampling family,
$$
Y, Z \mid \mu \overset{ind}{\sim} \textsc{nef-chs}(\mu, n).
$$
Further, we assume $\mu$ has a Pearson IV prior distribution with parameters $\mu_0$ and $m_0$,
$$
\mu \sim \text{Pearson IV}(\mu_0, m_0).
$$
Note that $\mu$ is a parameter in the sampling family and a random variable in the prior distribution. This is a standard pattern for Bayesian statistical models.

Letting $\mu_1 = (m_0\, \mu_0 + y) / (m_0 + n)$ and $m_1 = m_0 + n$, Bayes' Theorem enables us to compute the posterior distribution of $\mu$ given $Y = y$,
$$
(\mu \mid Y = y) \sim \text{Pearson IV}(\mu_1, m_1),
$$
with mean $\EXP \{\mu \mid Y=y \} = \mu_1$ and variance $\mathbb{VAR}(\mu \mid Y=y) = (\mu_1^2 + 1)/(m_1 - 1)$. Because the prior and posterior both belong to the Pearson IV family, it is called the conjugate family for \textsc{nef-chs} sampling.

The prior predictive distribution of $Y$ has density
$$
f_Y(y) = H(y, n)\, \frac{K(\mu_0, m_0)}{K(\mu_1, m_1)},
$$
with mean and variance
$$
\EXP\{Y\} = n\, \mu_0 \ \text{ and }
\Var\{Y\} = n\, (\mu_0^2 + 1)\, \frac{m_0 + n}{m_0 - 1}.
$$
We call this the ``Pearson IV--\textsc{nef-chs}$(n, \mu_0, m_0)$'' distribution, in analogy with the name beta-binomial for a beta mixture of binomials, and we write
$$
Y \sim \text{Pearson IV--}\textsc{nef-chs}(n, \mu_0, m_0).
$$

The posterior distribution is the reference distribution for estimating $\mu$. The prior predictive distribution is the reference distribution for model checking. See, for example, Box \cite{box1980sampling, box1983apology}. The posterior predictive distribution of $Z$ given $Y = y$ is
$$
(Z \mid Y = y) \sim \text{Pearson IV--}\textsc{nef-chs}(n, \mu_1, m_1).
$$
This is the reference distribution used to predict future observations.

Given this setup, values of $\mu$ from either the prior distribution or the posterior distribution can be generated using the algorithms developed in the earlier sections of this article. To generate values of $Y$ from the prior predictive distribution, we use a two-step process: first, (i) generate $\mu \sim \text{Pearson IV}(\mu_0, m_0)$, then (ii) generate $Y \mid \mu \sim \textsc{nef-chs}(\mu, n)$. To generate values of $Z$ from the posterior predictive distribution given $Y = y$, we use a similar two step process: first, (i) generate $\mu \sim \text{Pearson IV}(\mu_1, m_1)$, then (ii) generate $Z \mid \mu \sim \textsc{nef-chs}(\mu, n)$; note the change from $(\mu_0, m_0)$ to $(\mu_1, m_1)$.

\bibliographystyle{plainnat}
\bibliography{p.bib}
\end{document}